\documentclass[floatfix,amsmath,superscriptaddress]{revtex4}
\usepackage{amssymb}
\usepackage{graphicx}
\usepackage{graphics}
\usepackage{amsmath}
\usepackage{amsthm}
\usepackage{color}

\begin{document}

\theoremstyle{plain}
\newtheorem{theorem}{Theorem}

\newtheorem{lemma}[theorem]{Lemma}
\newtheorem{corollary}[theorem]{Corollary}
\newtheorem{conjecture}[theorem]{Conjecture}
\newtheorem{proposition}[theorem]{Proposition}

\theoremstyle{remark}
\newtheorem*{remark}{Remark}

\newcommand{\bea}{\begin{eqnarray}}
\newcommand{\eea}{\end{eqnarray}}
\newcommand{\e}{\eta}
\newcommand{\an}{\textbf{a}}
\newcommand{\bn}{\textbf{b}}
\newcommand{\cn}{\textbf{c}}
\def\bi{\begin{itemize}}
\def\ei{\end{itemize}}
\def\bc{\begin{center}}
\def\ec{\end{center}}
\def\E{{\cal E}}

\newcommand{\C}{\mathbb{C}}
\def\R{\hbox{$\mit I$\kern-.6em$\mit R$}}
\def\N{\hbox{$\mit I$\kern-.6em$\mit N$}}
\def\Y{Y}
\def\bk#1{\langle #1 \rangle}
\def\mw#1{\left< #1\right>}


\def\be{\begin{equation}}
\def\ee{\end{equation}}
\def\ba{\begin{align}}
\def\ea{\end{align}}

\newcommand{\mC}{\mathcal{C}}
\newcommand{\mE}{\mathcal{E}}
\newcommand{\mZ}{\mathcal{Z}}
\newcommand{\mU}{\mathcal{U}}
\newcommand{\mV}{\mathcal{V}}
\newcommand{\mA}{\mathcal{A}}
\newcommand{\mF}{\mathcal{F}}
\newcommand{\mI}{\mathcal{I}}
\newcommand{\mH}{\mathcal{H}}
\newcommand{\mL}{\mathcal{L}}
\newcommand{\mM}{\mathcal{M}}
\newcommand{\mT}{\mathcal{T}}
\newcommand{\mN}{\mathcal{N}}
\newcommand{\eqdef}{\equiv}

\newcommand{\fm}{\mathcal{F}_{\bf{m}}}
\newcommand{\am}{\mathcal{A}^{\textbf{m}}}
\newcommand{\dm}{\mathcal{D}(\mathrm{H}_{{\bf m}})}
\newcommand{\lr}{\rangle\langle}
\newcommand{\la}{\langle}
\newcommand{\ra}{\rangle}
\newcommand{\tr}{{\rm Tr}}

\newcommand{\mc}[1]{\mathcal{#1}}
\newcommand{\mbf}[1]{\mathbf{#1}}
\newcommand{\mbb}[1]{\mathbb{#1}}
\newcommand{\mrm}[1]{\mathrm{#1}}

\newcommand{\bra}[1]{\langle #1|}
\newcommand{\ket}[1]{|#1\rangle}
\newcommand{\braket}[3]{\langle #1|#2|#3\rangle}
\newcommand{\ip}[2]{\langle #1|#2\rangle}
\newcommand{\op}[2]{|#1\rangle \langle #2|}

\newcommand{\mbN}{\mathbb{N}}

\newcommand{\one}{\mbox{$1 \hspace{-1.0mm}  {\bf l}$}}

\definecolor{eric}{rgb}{0,.5,.2}
\newcommand{\eric}[1]{#1}

\newcommand{\review}[1]{{\color{red} #1}}

\title{Almost all multipartite qubit quantum states have trivial stabilizer}
\author{Gilad Gour}\email{giladgour@gmail.com}
\affiliation{
Department of Mathematics and Statistics,
University of Calgary, AB, Canada T2N 1N4}
\affiliation{
Institute for Quantum Science and Technology,
University of Calgary, AB, Canada T2N 1N4}
\author{Barbara Kraus}\email{Barbara.Kraus@uibk.ac.at}
\affiliation{Institute for Theoretical Physics, University of Innsbruck, Innsbruck, Austria}
\author{Nolan R. Wallach}\email{nwallach@ucsd.edu}
\affiliation{Department of Mathematics, University of California/San Diego,
        La Jolla, California 92093-0112}

\date{\today}


\begin{abstract}
The stabilizer group of an $n$-qubit state $|\psi\ra$ is the set of all matrices of the form $g=g_1\otimes\cdots\otimes g_n$,
with $g_1,...,g_n$ being any $2\times 2$ invertible complex matrices, that satisfy $g|\psi\ra=|\psi\ra$.
We show that for 5 or more qubits, except for a set of states of zero measure, the stabilizer group of multipartite entangled states is trivial; that is, containing only the identity element.
We use this result to show that for 5 or more qubits, the action of deterministic local operations and classical communication (LOCC) can almost always be simulated simply by local unitary (LU) operations. This proves that almost all $n$-qubit states with $n\geq 5$ can neither be reached nor converted into any other ($n$--partite entangled), LU--inequivalent state via deterministic LOCC. We also find a simple and elegant expression for the maximal probability to convert one multi-qubit entangled state to another for this generic set of states.
\end{abstract}

\maketitle

\section{Introduction}

Multipartite pure states lie at the heart of several fields of physics such as quantum information and condensed matter physics \cite{AmFa08,reviews}. Thereby, the local symmetries of the states often play a very important role. In condensed matter physics, the symmetries of a gapped Hamiltonian, and therefore of the ground state, are crucial in determining the phase in which the system is \cite{ScPe10}. Stabilizer states, which have applications in quantum error correction \cite{GoThesis,NiCh00}, measurement based quantum computation \cite{RaBr01}, and secret sharing \cite{secretsharing}, are in fact characterized by their local symmetries. In entanglement theory the local symmetries dictate which transformations among pure states are possible via local operation \cite{GoWa11,dVSp13}, which is, as we explain below, essential to reveal the entanglement properties of a state \cite{reviews}.

The identification of the local symmetries of a $n$--qubit state, $\ket{\Psi}$, i.e. the operators $g=g_1\otimes\ldots \otimes g_n$, with $g_i\in G$, where $G$ denotes a subgroup (or the whole group) of $\tilde{G}=GL(2)^{\otimes n}$, such that $g\ket{\Psi}=\ket{\Psi}$ is therefore an important task in all these fields. The subgroup of local symmetries, $G_\Psi$, is called stabilizer of the state $\ket{\Psi}$ (in $G$). We show here that for almost all states, i.e. up to a zero--measure set, of more than $4$ qubits, the stabilizer (in $\tilde{G}$) is trivial, i.e. the only local symmetry of the state is the identity. In particular we show that the set of states with trivial stabilizer is of full measure, open, and dense in the Hilbert space of $n\geq 5$ qubits.

To explain the far reaching consequences of this result in the context of entanglement theory let us note that one of the main reasons why bipartite (pure state) entanglement is so well understood, is that all possible transformations via local operations assisted by classical communication (LOCC), among pure states could be easily characterized \cite{nielsen}. As  entanglement theory is a resource theory where the free operations are LOCC \cite{reviews}, which implies that a function, $E$, is an entanglement measure (for pure states) if $E(\ket{\Psi})\geq E(\ket{\Phi})$ whenever $\ket{\Psi}$ can be transformed into $\ket{\Phi}$ via LOCC, this characterization allowed for the derivation of simple necessary and sufficient conditions on entanglement measures \cite{reviews,nielsen}. In the multipartite case, however, the characterization of LOCC transformations is far from being that simple. In contrast to the bipartite case even infinitely many rounds of classical communication might be necessary to accomplish certain transformations \cite{chit11}. Due to the complicated structure of LOCC \cite{chitambar1}, several other classes of operations have been considered to gain a better understanding of multipartite entanglement. These include Local Unitary (LU) operations, which do not alter the entanglement contained in the state as they can be applied and inverted locally \cite{Kr10}. Moreover, separable operations (SEP) \cite{rains}, which are strictly larger than LOCC \cite{sepnotlocc}, and stochastic LOCC (SLOCC) \cite{slocc}, i.e. probabilistic LOCC transformations, have been investigated. Recent investigations of SEP, showed that the existence of a transformation among two pure states depends strongly on the stabilizer of the state \cite{GoWa11}. The main result presented here, namely the fact that for almost all $n$--qubit states (with $n\geq 5$) the stabilizer is trivial implies then (see Theorem 6), that almost no pure state transformation via LOCC (and even via SEP) is possible. That is, almost all states are isolated, i.e. they can neither be transformed into another LU--inequivalent state, nor can they be reached from an LU--inequivalent state \footnote{Note that this fact has been proven for four qubits and three three--level systems using completely different tools in \cite{dVSp13,SpdV16,HeSp15}.}. Hence, the Maximally Entangled Set (MES) \cite{dVSp13}, which can be viewed as the generalization of the maximally entangled bipartite state, i.e. the minimal set required to generate an arbitrary $n$--qubit entangled state (for $n>4$) via LOCC, is of full measure in the Hilbert space. Furthermore, the states which are most relevant in the context of state manipulations via local operations are of measure zero. That is, as expected, the most powerful states in this context are very rare. Given the fact that deterministic transformations are impossible for generic states, we investigate also probabilistic transformations. In particular, we derive a very simple expression of the maximal success probability with which one state can be transformed into the other for this generic set of states.

We now provide more details of the content of this paper. First, we introduce our notation and recall some important results about Lie groups and algebraic geometry, such as the principal orbit type theorem (Theorem 1) and the Kempf--Ness theorem (Theorem 2). The latter characterizes critical states, i.e. states whose single qubit reduced states are all proportional to the identity. The principal orbit type theorem implies, in particular, that for a compact Lie group, such as the local unitaries, if there exists an $n$--qubit state whose unitary stabilizer is trivial then the set of states which have a trivial unitary stabilizer is open and dense in the whole Hilbert space of $n$-qubits, ${\cal H}_n$. We therefore proceed by first presenting an example of a $n$-qubit state, $\ket{L_n}$, which is critical and has a trivial stabilizer (in $G$) for any $n>4$ (Lemma 3). Then, using the Kempf-Ness theorem and other techniques from algebraic geometry, we show that the principal orbit type theorem can be used to prove that the set of states with trivial stabilizer in $\tilde{G}$ is of full measure, open, and dense in ${\cal H}_n$ (Theorem 4).

We next present the consequences of this result in entanglement theory.
There, we show that non--trivial (i.e. not Local Unitary (LU)) transformations among states with trivial stabilizers are impossible (Theorem 6). Hence, there is almost no transformation possible among entangled $n$--qubit states with $n>4$. Given that deterministic transformations are almost never possible we then study probabilistic transformations and present a simple formula for the maximal success probability with which one can transform one generic state, i.e. a state with trivial stabilizer, into the other via LOCC (Theorem 7).

\section{Preliminaries}

Let us start by introducing our notation and recalling some powerful theorems about Lie groups. We denote by  $\mH_{n}\eqdef\otimes^{n}\mathbb{C}^{2}$ the Hilbert space of $n$ qubits (we will use the symbol $\eqdef$ for an equality that represents a definition). Given a group $G \subset GL(\mH_{n})$  the \emph{stabilizer} of a state $|\psi\ra\in\mH_n$ with respect to this group (also known as an isotropy subgroup~\cite{Bredon}) is
\be
G_\psi\eqdef \left\{g\in G\;\Big|\;g|\psi\ra=|\psi\ra\right\}\subset G
\ee
Two stabilizer subgroups $G_\psi$ and $G_\phi$ are said to have the same \emph{type} if there exists a
$g\in G$ such that $G_\phi=g G_\psi g^{-1}$. That is, two groups are of the same type if they are conjugate in $G$.
Note that
$
gG_\psi g^{-1}=G_{g\psi}
$
for any $g\in G$ and $|\psi\ra\in\mH_n$. Therefore, $G_\psi$ and $G_\phi$ have the same type if and only if there exists $g\in G$ such that $G_\phi=G_{g\psi}$. Given a state $|\psi\rangle\in\mH_{n}$ we denote by
\be
G|\psi\rangle\eqdef\left\{g|\psi\rangle\;\Big|\;g\in G\right\}
\ee
the \emph{orbit} of $|\psi\ra$ under the action of $G$ (note that the orbit contains states that are not necessarily normalized).
Any orbit $G|\psi\ra$ is an embedded submanifold of $\mH_n$ and hence has a dimension.
Moreover, $G|\psi\ra$ is isomorphic to the left coset of $G_\psi$ in $G$, namely $G|\psi\ra \cong G/G_\psi$.
We therefore say that the orbits $G|\psi\ra$ and $G|\phi\ra$ have the same type if $G_\psi$ and $G_\phi$ have the same type. In addition, we say that an orbit $G|\psi\ra$ is of a lower type than $G|\phi\ra$, denoting it by
$
G|\psi\ra\prec_{type} G|\phi\ra\;\text{or equivalently by}\; G/G_\psi\prec_{type} G/G_{\phi}
$,
if $G_\phi$ is conjugate in $G$ to a subgroup of $G_\psi$.

An example of orbits which are of the same type would be given by two generic seed states of different SLOCC--classes of four qubits, denoted by $\ket{\Psi(a,b,c,d)}$ \cite{VeDe02}. The stabilizer of these states is (up to LUs) $\{\sigma_x^{\otimes 4},\sigma_y^{\otimes 4},\sigma_z^{\otimes 4},I\}$, where here, and in the following, $\sigma_x,\sigma_y,\sigma_z$ denote the Pauli operators and $I$ the identity operator. Hence, the stabilizers, $G_{\Psi(a,b,c,d)}, G_{\Psi(a',b',c',d')}$, and the orbits $G\ket{\Psi(a,b,c,d)}, G\ket{\Psi(a',b',c',d')}$ are of the same type (for any choice of $G$). Note that if the stabilizer of a state is trivial, i.e. $G_\psi=\{I\}$, it is always of maximal type, as $I$ is contained in any subgroup of $G$ (in other words $G/G_\psi\prec_{type} G/\{I\}$ for any $\psi\in\mH_n$).

The following fundamental theorem, from the field of Lie groups, shows the existence of a maximal element (known as the principal orbit type) under the preorder $\prec_{type}$.  This key theorem can be found in \cite{Bredon}, as a combination of Theorem 3.1 and Theorem 3.8.

\begin{theorem}{\rm \cite{Bredon}}\label{principal} {\rm \bf The principal orbit type theorem}
Let $C$ be a compact Lie group acting differentiably on a connected smooth
manifold $\mM$ (in this paper we assume $\mM\subset\mH_n$).
Then, there exists a principal orbit type; that is, there exists a state $\phi\in \mM$ such that $C/C_\psi\prec_{type} C/C_{\phi}$ for all $|\psi\ra\in \mM$.
Furthermore, the set of $|\psi\ra\in \mM$ such that $C_{\psi}$ is conjugate to $C_\phi$ is open
and dense in $\mM$ with complement of lower dimension and hence of measure 0.
\end{theorem}

In this paper we will consider the following 4 different groups all acting
on $\mH_n$:
\begin{align}
& G\eqdef SL(2,\mathbb{C})\otimes\cdots\otimes
SL(2,\mathbb{C})\subset SL(\mH_{n})\label{g}\\
& K\eqdef SU(2)\otimes\cdots\otimes
SU(2)\subset SU(\mH_{n})\label{k}\\
& \tilde{G}\eqdef GL(2,\mathbb{C})\otimes\cdots\otimes
GL(2,\mathbb{C})\subset GL(\mH_{n})\label{gl}\\
& \tilde{K}\eqdef U(2)\otimes\cdots\otimes
U(2)\subset U(\mH_{n})\label{u}
\end{align}
From here on $G,K,\tilde{G}$, and $\tilde{K}$, will always refer to these four groups. Let us now present a very important set of states, the so--called critical states. Let $Lie(G)$ be the Lie algebra of $G$ contained in $End(\mathcal{H}_{n})$. The set of \emph{critical} states is defined as
$$
Crit(\mathcal{H}_{n})\equiv\{\phi
\in\mathcal{H}_{n}|\left\langle \phi|X|\phi\right\rangle =0,X\in
Lie(G)\}.
$$
Many important states in quantum information such as the Bell states, GHZ states~\cite{ghz}, cluster states~\cite{Rau01}, graph states~\cite{Hei05}, and code states~\cite{NiCh00} are all critical. Moreover, as discussed in~\cite{GW13},
the orbit $G\cdot Crit(\mH_n)$ (i.e. the union of all orbits containing a critical state)
is open, dense, and of full measure in $\mH_n$.
The  following theorem provides a fundamental characterization of critical states.
\begin{theorem}\label{KN}{\rm  \bf The Kempf-Ness theorem~\cite{KN}}
\begin{enumerate}
\item $\phi\in Crit(\mathcal{H}_{n})$ if and
only if $\left\Vert g\phi\right\Vert \geq\left\Vert \phi\right\Vert $ for all
$g\in G$, where here and in the following $\left\Vert \phi \right\Vert$ denotes the Euclidean norm of $\ket{\phi}\in \mathcal{H}_{n}$.
\item If $\phi\in Crit(\mathcal{H}_{n})$ and $g\in G$ then $\left\Vert g\phi\right\Vert
\geq\left\Vert \phi\right\Vert $ with equality if and only if $g\phi\in K\phi$.
Moreover, if $g$ is positive definite then the equality condition holds if and only if
$g\phi=\phi$.
\item If $\phi\in\mathcal{H}_{n}$ then $G\phi$ is closed in $\mathcal{H}_{n}$ if
and only if $G\phi\cap Crit(\mathcal{H}_{n})\neq\emptyset$.
\end{enumerate}
\end{theorem}
From this theorem it follows~\cite{Ver03,GW10} that
$\phi\in\mathcal{H}_n$ is critical
if and only if all the \emph{local} density matrices are proportional to the identity (i.e.
each qubit is maximally entangled with the remaining $n-1$ qubits). Using this characterization of critical states it follows easily that the states mentioned above are all critical. Moreover, part 2 of the Kempf-Ness theorem above implies the uniqueness of critical states; that is, up to LUs there can be at most one critical state in any invertible SLOCC orbit. That is, critical states are the natural representatives of SLOCC orbits, and are the unique states in their SLOCC orbits for which each qubit is maximally entangled to the other qubits ~\cite{GoWa11}. Moreover, all these states are in the MES~\cite{dVSp13}. Another important property of a critical state, $\phi$, is that its stabilizer, $\tilde{G}_\phi$ is symmetric. That is if $g\in \tilde{G}_\phi$ then $g^\dagger \in\tilde{G}_\phi$ for any critical state, $\phi$. We give the proof of this result in Appendix A.

Let us now also introduce the notion of SL-invariant polynomials (SLIPs). A polynomial $f:\mH_n\to \mathbb{C}$ is said to be SL-invariant if $f(g|\psi\ra)=f(|\psi\ra)$ for all $g\in G$ and $|\psi\ra\in\mH_n$. SLIPs can be used to characterize the (generic) orbits of $G$ and have been studied extensively in literature (see, e.g.~\cite{OsSi12, GW13} and references therein). For $|\psi\ra\in\mH_n$ the polynomial $f_2(\psi)\eqdef \la\psi^{*}|\sigma_{y}^{\otimes n}|\psi\ra$, for instance, is homogeneous SLIP of degree 2. However, if $n$ is odd $f_2(\psi)=0$ for all $\psi$. In fact, it is know that for odd number of qubits the lowest degree of a non-trivial homogeneous SLIP is 4. An example of such a SLIP of degree 4 is ~\cite{GW13}
\be
f_{4}(\psi)\eqdef\det\begin{pmatrix}
\bra{\varphi_0^\ast}\sigma_y^{\otimes n-1}\ket{\varphi_0} & \bra{\varphi_0^\ast}\sigma_y^{\otimes n-1}\ket{\varphi_1}\\
\bra{\varphi_1^\ast}\sigma_y^{\otimes n-1}\ket{\varphi_0}) & \bra{\varphi_1^\ast}\sigma_y^{\otimes n-1}\ket{\varphi_1}
\end{pmatrix}.\label{f4}
\ee
Here, $|\psi\ra=|0\ra|\varphi_{0}\ra+|1\ra|\varphi_1\ra$ where $|\varphi_{0}\ra$ and $|\varphi_1\ra$ are (non-normalized) states in $\mH_{n-1}$ with $n$ odd.

\section{Existence of a state with trivial stabilizer}

Before stating our main result, we now show that there exists one state with a trivial stabilizer in $G$ (and therefore also in $K$).
We define a normalized state $|L_n\ra\in \mH_n$ as follows:
\be
\label{eq_Lstate}
|L_{n}\ra\eqdef\frac{1}{\sqrt{2(n-1)}}\left(\sqrt{n-2}\left\vert 11...1\right\rangle +\sqrt{n}|W_n\ra\right)\;,
\ee
where $|W_n\ra\eqdef\frac{1}{\sqrt{n}}\left(|10...0\ra+|01...0\ra+\cdots+|0...01\ra\right)$. The state $|L_n\ra$ has several interesting properties. It is symmetric under permutations of the qubits,
and the reduced density matrix of each qubit is proportional to the identity. Hence, $|L_n\ra\in Crit(\mH_n)$ is a critical state. Note that $|L_2\ra$ is a Bell state,
$|L_3\ra$ is the GHZ state (up to local unitaries), and $|L_4\ra=|L\ra$, where $|L\ra$ was originally defined on 4 qubits in~\cite{Ost2005,Love} and was shown in~\cite{GW10} to have many interesting properties, particularly identifying it as the 4 qubit state maximizing many important entanglement measures and being in the MES \cite{SpdV16}.

Using the SLIPs, $f_2$ and $f_4$ introduced above, and the Kempf-Ness theorem, we prove that the stabilizer $G_{\ket{L_n}}$ (and therefore also $K_{\ket{L_n}}$) is trivial for $n>4$, as stated in the following lemma (see Appendix A for the proof)
\begin{lemma}
\label{trivialstab}
If $n>4$ then the stabilizer of $|L_{n}\ra$ in $G$ is
trivial. That is, if $g\in G$ and $g|L_n\ra=|L_{n}\ra$ then $g=I$, where $I$ is the identity element of $G$.
\end{lemma}

\section{Genericity of trivial stabilizer for $5$ or more qubits.}

We are now ready to present our main result.
\begin{theorem}\label{mtheorem}
For $n\geq 5$ there exists a subset of states $\mA\subset\mH_n$ that is open, dense, and of full measure in $\mH_n$,
such that the stabilizer group $\tilde{G}_{\psi}=\{I\}$ is trivial for all $|\psi\ra\in\mA$.
\end{theorem}

Before we outline the proof of the theorem (for details of the proof see Appendix C), we point out that for the compact Lie group $K$ the result already follows from Lemma \ref{trivialstab} and the principle orbit type theorem. To see that, apply Theorem~\ref{principal} with $C=K$ and $\mM=\mH_{n}$. The theorem first states that there exists a state $|\phi\ra\in\mH_n$ such that $K|\phi\ra$ is a principle orbit type. However, we know this already from Lemma \ref{trivialstab}; simply take $|\phi\ra=|L_n\ra$. Due to Lemma~\ref{trivialstab} we have  $K_{\phi}=\{I\}$. Therefore, $K_{\phi}$ is a subgroup of $K_\psi$ for all $|\psi\ra\in\mH_n$. That is, $K|\psi\ra\prec_{type}K|\phi\ra$ for all $|\psi\ra\in\mH_n$. Hence, $K|\phi\ra$ is a principle orbit type. The second part of Theorem~\ref{principal} states that the set of states $\mA$, such that $K|\psi\ra$ is a maximal orbit type for all $|\psi\ra\in\mA$ (or equivalently $\mA$ is the set of all states with trivial stabilizer in $K$), is open, dense, and of full measure in $\mH_n$.

Unlike the group $K$ (or $\tilde{K}$), the groups $G$ and $\tilde{G}$ are not compact, and therefore the principle orbit type theorem cannot be applied directly in these cases. For this reason, we first define the following set of states in $\mH_n$:
\be\label{C0}
\mC\eqdef\left\{\psi\in Crit(\mH_n)\;\Big|\;\dim (G|\psi\ra)=\dim(G)\right\}
\ee
That is, $\mC$ consists of all critical states whose orbits (under $G$) have maximal dimension (i.e. the dimension of $G$). In particular, due to the identity $G|\psi\ra \cong G/G_\psi$ it follows that $|\psi\ra\in\mC$ if and only if $|\psi\ra$ is critical and $G_\psi$ is a finite group (or equivalently $\dim(G_{\psi})=0$). The set $\mC$ has several key properties which we summarize in the following lemma and which we need for the proof of Theorem~\ref{mtheorem}.

\begin{lemma}\label{properties}
The set $\mC$ defined in Eq.(\ref{C0}) has the following properties:
\begin{enumerate}
\item $G_\psi=K_\psi$ for all $|\psi\ra\in\mC$.
\item The set $G\mC \eqdef \left\{g|\psi\ra\;|\;g\in G\;;\;|\psi\ra\in\mC\right\}$ is open with complement of lower dimension in $\mH_n$.
\item $\mC$ is a connected smooth submanifold of $\mH_n$, and $K$ acts differentiably on $\mC$.
\end{enumerate}
\end{lemma}

Property~(1) was first proved in Proposition 5 of~\cite{GoWa11} and implies that when restricting ourselves
to $\mC$ we can restrict the analysis to the compact group $K$ instead of the non-compact group $G$.
Property~(2) is not trivial and implies that the set of all states that are \emph{not} in $G\mC$ is of measure zero.
Finally, Property~(3) can be divided into two parts:
the first one is that $\mC$ is a submanifold of $\mH$. This property follows essentially from the implicit function theorem of multi-variable calculus. The second part, that $\mC$ is connected follows from the connectedness of $K$.
With this Lemma at hand, we are now ready to outline the proof of Theorem~\ref{mtheorem}. We leave the details of the proof of Lemma~\ref{properties} and of Theorem~\ref{mtheorem} to the appendix (see Appendix B and C).

{\it Outline of the proof of Theorem~\ref{mtheorem}:}
We first show that there exists an open set, with complement of lower dimension, $\mA'$ in $\mH_n$, such that for all $|\psi\ra\in\mA'$, $G_\psi=\{I\}$ is trivial, and then use it to show that there also exists a subset $\mA\subset\mH_n$ that is open, dense, and of full measure in $\mH_n$ for which
$\tilde{G}_\psi=\{I\}$ for all $|\psi\ra\in\mA$.

In order to do so, we first apply the principle orbit type theorem to $\mM$ being the connected smooth submanifold $\mC\subset\mH_n$ (see property~3 of Lemma~\ref{properties}),
and the compact Lie group $C=K$. Due to the existence of a state $\ket{L_n} \in \mC$ which has trivial stabilizer in $G$ (and therefore in $K$), and due to properties~1 and 2 of Lemma~\ref{properties}, we have that the set $\mC'= \left\{|\psi\ra\in\mC\;\big|\;G_\psi=\{I\}\right\}$ is open and dense in $\mC$.
The set $\mA'\eqdef G\mC'=\left\{g|\psi\ra\;\big|\;|\psi\ra\in\mC'\;;\;g\in G\right\}$ is therefore open and dense in $G\mC$. Clearly, any state $\ket{\phi}$ in $G\mC$ has a trivial stabilizer. Using now that $G\mC$ is open with complement of lower dimension in $\mH_n$ (see property~2 in Lemma~\ref{properties}), we also have that $\mA'$, which contains only states with $G_\psi=\{I\}$ is open with complement of lower dimension in $\mH_n$. In order to define the set $\mA$ containing states with trivial stabilizer in $\tilde{G}$ (not only $G$) which is also open and with complement of lower dimension in $\mH_n$, one can use the SLIPs introduced above. With these SLIPs we are able to identify a subset $\mA\subset\mA'$ with the desired properties, such that for any state $\ket{\phi}\in \mA$, $g\in \tilde{G}_\phi$ only if $g\in G_\phi$ (and therefore trivial). The set $\mA$ has now the desired properties, which completes the proof.

Let us now discuss the consequences of Theorem~\ref{mtheorem} in the context of entanglement theory. We will call a $n$--qubit state `$n$-way entangled' if it cannot be written as a product state between one qubit and the rest of the $n-1$ qubits. Regarding pure state transformations, Theorem~\ref{mtheorem} implies the following theorem.
\begin{theorem} \label{ThnoDet}
Let $|\psi\ra,|\phi\ra\in\mH_n$ be two $n$-way entangled states and suppose also that the stabilizer $\tilde{G}_\psi=\{I\}$ is trivial. Then, $|\psi\ra$ can be converted deterministically to $|\phi\ra$ by LOCC or SEP operations if and only if $|\psi\ra$ and $|\phi\ra$ are LU--equivalent, i.e. there exists $u\in \tilde{K}$ such that
$|\psi\ra=u|\phi\ra$.
\end{theorem}

The proof of this theorem is straightforward and is presented in Appendix D. Hence, under both SEP and LOCC, deterministic transformations of the form $|\psi\ra\to|\phi\ra$ are not possible for two states in $\mA$ unless $|\psi\ra$ and $|\phi\ra$ are LU--equivalent. It is therefore crucial to determine what is the maximum possible probability with which the transformation $|\psi\ra\to|\phi\ra$ can be achieved by LOCC, which is stated in the following theorem.
\begin{theorem} \label{ThmaxProb}
Let $|\psi\ra\in\mH_n$ be a normalized $n$-way entangled state with $\tilde{G}_\psi=\{I\}$. Let $|\phi\ra=g|\psi\ra$ be a normalized state in the SLOCC orbit $\tilde{G}|\psi\ra$. Then, the maximum probability with which $|\psi\ra$ can be converted to $|\phi\ra$ by LOCC or SEP is given by:
\be\label{pmax1}
p_{\max}(|\psi\ra\to|\phi\ra)=\frac{1}{\lambda_{\max}(g^{\dag}g)},
\ee
where $\lambda_{\max}(X)$ denotes the maximal eigenvalue of $X$.
\end{theorem}
From the proof of this theorem, which is presented in Appendix D, it will also follow that the RHS of~\eqref{pmax1} always provides a lower bound on the maximal probability even if $|\psi\ra\notin\mA$. However, if the stabilizer of $|\psi\ra$ is not trivial, $g$ is not unique, and the lower bound given by the RHS of~\eqref{pmax1} can be improved by maximizing it over all $g\in\tilde{G}$ that satisfies $|\phi\ra=g|\psi\ra$.

\section{Conclusions}
In summary, we have shown that the set of states describing more than four qubits with trivial stabilizer is dense, and of full measure in the Hilbert space. That is, for almost every $n$--qubit state with $n>4$ there exists no non--trivial local symmetry of the state. We used this result to prove that among states in this full measure set there is no non-trivial transformation possible, i.e. other then LUs. Hence, these states can only be transformed probabilistically into each other via local operations. We also determined the maximal success probability for these transformations, which can be easily computed.

Note that these results also imply that there are only very mild conditions on a function to be an entanglement measure of pure states, as any such function only has to obey that it is non--increasing under pure state LOCC--transformations. Furthermore, the MES of $n$--qubit states with $n>4$ is of full measure. The identification and characterization of convertible states, as well as the higher dimensional case, will be part of a future work. We believe that the results and the tools presented here might be very useful to investigate related problems and might also lead to important consequences in other fields of physics.

\section{Acknowledgments}   G.G. research is supported by NSERC. B. K is acknowledging support of the Austrian Science Fund (FWF): Y535-N16.  

{\it Note added-} We became aware after the publication of this paper that the state $|L_n\rangle$ appeared first in~\cite{Ost2005}.

\begin{appendix}

\section*{Appendix}

We present here all the detailed proofs of the theorems and lemmata presented in the main text. In Appendix A we analyze the stabilizer in $\tilde{G}$ of critical states and show that the stabilizer (in $G$) of the $\ket{L_n}$ state is trivial. In Appendix B we present the proof of Lemma \ref{properties}, where the properties of the set $\mC$ of critical states with finite stabilizer (in $G$) are stated. This lemma is then used in Appendix C to prove the main theorem of the paper, namely that almost all $n$--qubit states with $n>4$ have a trivial stabilizer (in $\tilde{G}$). In Appendix D we prove in detail the direct consequences of the main theorem of the paper in the context of multipartite entanglement theory.

\section*{Appendix A: Stabilizers of critical states and the $\ket{L_n}$ state}
\label{sectionA}

In this appendix we investigate the stabilizer of critical states and of the $\ket{L_n}$ state. We first show that the stabilizer of any critical state is symmetric, that is if $g\in \tilde{G}_\psi$ then $g^\dagger \in \tilde{G}_\psi$ (see Theorem \ref{kncor} below). We then use this result to show that the stabilizer of the critical $\ket{L_n}$ state presented in the main text (see Eq. (\ref{eq_Lstate})) is trivial.

\subsection*{The stabilizer of critical states}

The following theorem is a consequence of the Kemp-Ness theorem on the stabilizer subgroups of $\tilde{G}$.
\begin{theorem}\label{kncor}
Let $|\psi\ra\in Crit(\mH_n)$ be any critical state, and let $g\in\tilde{G}_{\psi}$. Then, also $g^{\dag}\in\tilde{G}_\psi$.
That is,  $\tilde{G}_\psi$ is invariant under the adjoint.
\end{theorem}
\begin{proof}
The Hilbert-Mumford theorem (see for example~\cite{Sh94} or in the context of our paper see~\cite{Wa16}) implies that if $|\psi\ra$ is critical, then there exists at least one homogeneous SLIP $f$ of some degree $m$ such that $f(|\psi\ra)\neq 0$. Now, if $g\in\tilde{G}_{\psi}$ we can write it as $g=tg'$, where $0\neq t\in\mathbb{C}$ and $g'\in G$. We therefore have $f(|\psi\ra)=f(g|\psi\ra)=t^mf(g'|\psi)=t^mf(|\psi\ra)$. Since $f(|\psi\ra)\neq 0$ this gives $t^m=1$. Next, using the polar decomposition we have $g=u\sqrt{g^{\dag}g}$ with $u\in\tilde{K}$. Hence,
$\||\psi\ra\|=\|g|\psi\ra\|=\|\sqrt{g^\dag g}|\psi\ra\|$. Moreover, note that $\sqrt{g^{\dag}g}\in G$ since $t$ is just a phase. Therefore, from the second part of the Kemp-Ness theorem above it follows that $\sqrt{g^\dag g}|\psi\ra=|\psi\ra$. Applying again $\sqrt{g^\dag g}$ to both sides of this relation gives $g^\dag g|\psi\ra=|\psi\ra$.
But since $g|\psi\ra=|\psi\ra$ we conclude that $g^\dag|\psi\ra=|\psi\ra$.
\end{proof}

\subsection*{Trivial stabilizer of $\ket{L_n}$}

In this subsection we present the proof of Lemma \ref{trivialstab}, which we restate here in order to improve readability.

{\bf Lemma 3.}
{\it If $n>4$ then the stabilizer of $|L_{n}\ra$ in $G$ is
trivial. That is, if $g\in G$ and $g|L_n\ra=|L_{n}\ra$ then $g=I$, where $I$ is the identity element of $G$.}

\begin{proof}
It will be convenient to work with the unnormalized states $|\ell_n\ra\eqdef\sqrt{2(n-1)}|L_n\ra$ and $|w_n\ra\eqdef\sqrt{n}|W_n\ra$. Note that with these definitions
\be\label{ln}
|\ell_n\ra=|0\ra|w_{n-1}\ra+|1\ra\Big(\sqrt{n-2}|1...1\ra+|0...0\ra\Big).
\ee
Let $g\in G$ and denote
$g=g_{1}\otimes h$, where $g_1\in SL(2,\mathbb{C})$ acts on the first qubit, and $h=g_2\otimes\cdots\otimes g_n$ acts on the remaining $n-1$ qubits. Denote further
\be\label{g111}
g_1=\begin{pmatrix}
a & b \\
c & d
\end{pmatrix}\;.
\ee
Applying $\la 0|\otimes I$ on both sides of the relation $g_1\otimes h|\ell_n\ra=|\ell_n\ra$ gives
\be
|w_{n-1}\ra=h\left[a|w_{n-1}\ra+b\Big(\sqrt{n-2}|1...1\ra+|0...0\ra\Big)\right]\label{f1}
\ee
where we have used Eqs.~(\ref{ln},\ref{g111}). We now show that the above equation implies that $b=0$.

Suppose first that $n\geq 5$ is odd. In this case, $n-1$ is even so we can apply the SLIP $f_2$ on both sides
of the equations above. Applying it to~\eqref{f1}, gives
\be
0=f_{2}\left(h\left[a|w_{n-1}\ra+b\Big(\sqrt{n-2}|1...1\ra+|0...0\ra\Big)\right]\right)
\ee
where we have used the fact that $f_2(|w_{n-1}\ra)=0$. Since $\det(h)=1$ we have
\begin{align}
0&=f_{2}\left(a|w_{n-1}\ra+b\Big(\sqrt{n-2}|1...1\ra+|0...0\ra\Big)\label{wn-1}\right)\nonumber\\
& = 2b\sqrt{n-2}\;,
\end{align}
where the last equality follows from the definition of $f_2$ in terms of the bilinear form and the fact that for $n>3$ we have $\bra{w_{n-1}^\ast}\sigma_y^{\otimes n-1}\ket{w_{n-1}}=\bra{w_{n-1}^\ast}\sigma_y^{\otimes n-1}\ket{0...0}=\bra{w_{n-1}^\ast}\sigma_y^{\otimes n-1}\ket{1...1}=0$.
We therefore conclude that $b=0$.

Suppose now that $n\geq 6$ is even. In this case $n-1$ is odd, so that we can apply the SLIP $f_4$ (see Eq.~\eqref{f4}) to both sides of~\eqref{f1}. Since $f_4(|w_{n-1}\ra)=0$ and recalling that $\det(h)=1$ we obtain
\begin{align}
0&=f_4\left(a|w_{n-1}\ra+b\Big(\sqrt{n-2}|1...1\ra+|0...0\ra\Big)\right)\nonumber\\
& = -(n-2)b^4\;,
\end{align}
where the last equality follows from the definition given in~\eqref{f4} and the fact that $n>4$. Hence, $b=0$ also in this case.

We therefore showed that if $n\geq 5$ then $b=0$. Similarly, since we also have $g^{\dag}|L_n\ra=|L_n\ra$ (see Theorem~\ref{kncor}) we also get $c=0$. Therefore, since $|L_n\ra$ is invariant under permutations we conclude that $g=g_1\otimes g_2\otimes\cdots\otimes g_n$  with all $g_i$ with $i=1,2,...,n$ diagonal. Finally, using the notation:
\be\label{g1}
g_i=\begin{pmatrix}
a_i & 0 \\
0 & d_i
\end{pmatrix}\quad i=1,...,n\;,
\ee
we get
\begin{align}
&|\ell_n\ra =g|\ell_n\ra=\sqrt{n-2}d_1\cdots d_n|1...1\ra\nonumber\\
&\;\;+d_1a_2\cdots a_n|10...0\ra+\cdots+a_1\cdots a_{n-1}d_{n}|0...01\ra
\end{align}
so $d_{1}d_{2}\cdots d_{n}=1$. Moreover, as $det(g_i)=1$ we obtain from $d_1a_2\cdots a_n=a_1 d_2 \cdots a_n$ etc. that $d_1^2=d_i^2$ for any $i=2,...,n$. Combining the equation $d_1a_2\cdots a_n=1$ with $d_{1}d_{2}\cdots d_{n}=1$ yields then $d_1^2=1$. Hence, $d_i^2=1$ and therefore $g=I$ as asserted.
\end{proof}

Note that this lemma can also be proven using the fact that $|L_n\ra$ is an entangled permutational-symmetric state of $n$-qubits, which is neither in the GHZ--class nor in the $W$-class.
For such states, the results presented in~\cite{MaKr10} imply that any local symmetry of it must be of the from $g^{\otimes n}$, with $g\in GL(2)$. Note however that in $\tilde{G}$ the stabilizer is not trivial. In particular, $g^{\otimes n}$, with $g$ being a diagonal matrix with entries $e^{i 2\pi/(n-1)},1$ is contained in $\tilde{G}_{|L_n\ra}$.

\section*{Appendix B: Proof of Lemma \ref{properties}}
\label{sectionB}

We present here the proof of the properties of the set $\mC$, i.e. the set of critical states having a finite stabilizer in $G$ (see Lemma \ref{properties}). Whereas property 1 of Lemma \ref{properties} has been proven in Proposition 5 of \cite{GoWa11}, it remains to prove property 2 and 3. In order to increase readability of this appendix, we state this two properties as two theorems here (see Theorem \ref{propertyii} and Theorem \ref{propertyiii} below).

Whereas we avoided to use the notion of Zariski topology in the main text, we are going to use it in this appendix.
A subset, $X
\subset \C^m$ is called Zariski closed if there exists a set of polynomials, $S\subset \C[x_1,\ldots,x_m]$ such that $X=\C^m(S)=\{x\in \C^m|s(x)=0,\forall s\in S\}$. Here, and in the following $\C[x_1,\ldots,x_n]$ denotes the polynomial ring. That is, Zariski closed sets are defined as the common zeros of a set of polynomials. It can be easily verified that Zariski closed sets obey the axioms of closed sets in a topology, i.e. the whole and the empty set are closed, the union of finitely many closed sets is closed and the intersection (also of infinitely many) closed sets is closed. Note that a Zariski closed set is also closed in the Euclidian topology in $\C^m$. Important facts in the context studied here are that a non--empty Zariski open set is always open and dense in the Euclidian topology on  $\C^m$. As the complement of a Zariski open set is  Zariski closed, any non--empty Zariski open set is of full measure. On the other hand, not every open and dense set is Zariski open.  An algebraic variety is a set of solutions to polynomial equations,
i.e. is $Z$-closed set. An algebraic subvariety is a subset of an algebraic variety which is itself a variety.

Let us also recall here the definition of a ring and an ideal. A ring, is a set $R$ equipped with the two binary operations $+$ and $\cdot$ such that $(R,+)$ is an abelian group and $(R,\cdot)$ is a monoid. A subset ${\cal I} \subset R$  is called an ideal of $R$ if $( {\cal I} , + )$ is a subgroup of $( R , + )$ and $ \forall a\in {\cal I},\forall r\in R$ it holds that $a\cdot r,r\cdot a\in {\cal I}$. To give an example, the set of polynomials vanishing on a subset of $X \subset \C^m$, ${\cal I}_X=\{f\in \C[x_1,\ldots,x_n]|f(x)=0 \forall x\in X\}$, forms an ideal in the polynomial ring $\C[x_1,\ldots,x_m]$. An ideal is called proper ideal if it is strictly smaller than $R$. If ${\cal I} \subset \C[x_1,\ldots, x_n]$ is an ideal then we set $\C^n[{\cal I}]=\{x\in \C^n| f(x)=0 \forall f\in {\cal I}\}$. As commonly used, we write $f_{|X}=0$ for $f$ such that $f(x)=0$ for any $x\in X$, e.g. ${\cal I}_X=\{f\in \C[x_1,\ldots,x_n]|f_{|X}=0 \}$. An ideal ${\cal I} \subset \C[x_1,\ldots, x_n]$ is called radical if whenever there exists a $k\geq 1$ such that $f^k\in {\cal I}$ then $f\in {\cal I}$. The radical, $\sqrt{{\cal I}}$, of an ideal ${\cal I}$ is then defined as $\sqrt{{\cal I}}\equiv\{f \in \C[x_1,\ldots,x_n] | f^k\in {\cal I} \mbox{ for some } k\}$. It is the smallest radical ideal which contains ${\cal I}$.

Hilbert's Nullstellensatz \cite{Hilbert}, which we are going to recall next (for $\C$), provides a fundamental correspondence between (radical) ideals of polynomial rings and algebraic varieties of $\C^{n}$.

\begin{theorem} (Hilbert's Nullstellensatz \cite{Hilbert})
Let ${\cal I} \subset \C[x_1,\ldots,x_n]$ be an ideal of $\C[x_1,\ldots,x_n]$. Then
\bi \item[1] $\C^n({\cal I})=\emptyset$ iff ${\cal I} = \C[x_1,\ldots,x_n]$.
\item[2] ${\cal I}_{\C^n({\cal I})}=\sqrt{{\cal I}}$.
\ei
\end{theorem}

The first statement means that for any proper ideal there exists a common zero for all the polynomials in the ideal, i.e. the algebraic set $\C^n({\cal I})=\{x \in \C^n | f(x)=0 \forall f\in {\cal I}\}$ is non--empty. Or stated differently the ideal ${\cal I}$ contains the constant polynomial $1$ iff there exists no common zero (in $\C^n$) to all the polynomials in ${\cal I}$. The second statement means that any function which vanishes on $\C^n({\cal I})$ is a root of a polynomial in ${\cal I}$.

Let us also introduce the Reynolds operator, $R$ (see for example~\cite{Wa16}). For a finite group, $H$ and $f\in \C[x_1,\ldots,x_n]$ we have
\bea R_H(f)\equiv \frac{1}{|H|} \sum_{h\in H} h \circ f,\eea
where $h \circ f(x_1,\ldots,x_n)=f(h^{-1}(x_1,\ldots,x_n))$. For a compact Lie group, such as $K$ one defines $R_H$ similarly, where the sum is replaced by an integral. It is easy to check that $R_H(f)$ is $H$--invariant, i.e. $R_H[f(\vec{x})]=R_H[f(h'\vec{x})]$ for $\vec{x}=(x_1,\ldots,x_n)$ and for any $h'\in H$. In order to explain how the Reynold operator is defined for a symmetric, non--compact subgroup of $GL(n)$, let us consider the example of $G$ and use the notation $N=2^n$. $R_G$ is defined as the integral over the compact Lie group $G\bigcap U(N)$. In the example considered here, $G\bigcap U(N)=K$. The fact that $R_G(f)$ is $G$--invariant for any $f\in \C[x_1,\ldots,x_N]$ can be shown as follows. Note that $G=K e^{i Lie(K)}$ (polar decomposition) and that, due to the construction $R_G[f(kx)]=R_G[f(x)]$ for any $k\in K$. Let us define $\mu(Y)=R_G[f(Y x)]-R_G[f(x)]$ for a fixed $x\in \C^N$, which is a polynomial on $M_N(\C)$. Then we have that $\mu(Y)=0$ for any $Y\in K$. Consider now the complex analytic function $\Phi(z)=\mu(k e^{z X})$, where $k\in K, X\in Lie(K)$ are fixed and $z\in \C$. Since $\Phi(t)=0$ for any $t\in \R$ and since a complex analytic function on $\C$ is completely determined by its values on $\R$ we have $\Phi(z)=0$ $\forall z\in \C$. Hence, $\mu(Y)=0$ for any $Y\in G$ and therefore $R_G[f(Y x)]=R_G[f(x)]$ for any $Y\in G$. Thus, we have constructed the Reynolds operator $R_G$, which maps polynomials to $G$--invariant polynomials. We will use this operator in the proof of the following lemma (see \cite{Wa16}), which we need to prove Property~2 of Lemma~\ref{properties}.

\begin{lemma}
Let $A$ and $B$ be two subvariaties of $\mH_n$ that are Zariski closed in $\mH_n$ and are $G$-invariant; that is, if $|\psi\ra\in A$ then $G|\psi\ra\subset A$ (the same holds for $B$). Suppose also that $A\cap B=\emptyset$. Then, there exists a SLIP $f$ such that $f|_{A}=0$ and $f|_{B}=1$.
\end{lemma}

\begin{proof}

Let ${\cal I}$ denote the ideal of polynomials in $\C[x_1,\ldots,x_m]$ which vanish on $A$, i.e. ${\cal I}_A$. Then $L={\cal I}_{|B}$ is an ideal in ${\cal O}(B)=\C[x_1,\ldots,x_n]_{|B}$. As $A \bigcap B=\emptyset$ there exists no $b\in B$ such that all polynomials in $L$ vanish on $b$. Due to Hilbert's Nullstellensatz we have that $L={\cal O}(B)$. In particular, $L$ contains the constant polynomial $1$. Let $f'$ denote the polynomial in ${\cal I}$ such that $f'_{|B}=1$. Hence, $f'$ has the required property that it vanishes on $A$ and is constantly $1$ on $B$. In order to construct now a function which is $G$--invariant, we use the Reynolds operator, $R_G$ explained above, which maps polynomials into $G$--invariant polynomials and define $f=R_G(f')$. In order to complete the proof now, it remains to show that also the $G$--invariant polynomial $f$ fulfills the conditions that $f|_{A}=0$ and $f|_{B}=1$. As shown in \cite{Wa16} (see Corollary 15), for a $G$--invariant and Zariski closed subset, $X$, it holds that $R_G(f')|_{X}=R_G(f'_{|X})$. As both, $A$ and $B$ are Zariski closed and $G$--invariant this implies that $f_{|B}=R_G(f'_{|B})=1$ and $f_{|A}=R_G(f'_{|A})=0$, which completes the proof.
\end{proof}

We use this lemma now in order to prove property 2 of Lemma~\ref{properties}, which we state as the following theorem.

\begin{theorem}
\label{propertyii}
The set $G\mC\eqdef\{g|\psi\ra\;|\;g\in G\;\;;\;\;|\psi\ra\in\mC\}$ is open with a complement of lower dimension in $\mH_n$.
\end{theorem}
\begin{proof}
For $n=2,3,4$ this result has been proven in \cite{GoWa11}. Hence, it remains to prove it for any $n\geq 5$. To do so, we denote by $d\eqdef\dim(G)$ and by
\be
\mZ\eqdef \left\{|\psi\ra\in\mH_n\;\Big|\;\dim\left(G|\psi\ra\right)<d\right\}\;.
\ee
Then, by definition, $|\psi\ra\in\mZ$ if and only if for any $d$ elements $\{X_j\}_{j=1,...,d}$ in $Lie(G)$ the vectors
$\{X_j|\psi\ra\}_{j=1,...,d}$ are linearly dependent. That is, $|\psi\ra\in\mZ$ if and only if
\be
X_1|\psi\ra\wedge X_2|\psi\ra\wedge\cdots\wedge X_d|\psi\ra=0\quad\forall\;X_j\in Lie(G)\;,
\ee
with $j=1,2,...d$. The condition above can be expressed in terms of equations of the form
$f(|\psi\ra)=0$, where $f$ is polynomial, and we conclude that the set $\mZ$ is Zariski closed in $\mH_n$.
Its complement $\mH_n-\mZ$ is therefore Zariski open.

Now, let
$|\phi\ra\in\mH_n$ be a state such that $G|\phi\ra$ is Zariski closed and of dimension $d$ (i.e. maximal dimension). Note that such a state exists, as the critical state $\ket{L_n}$ (see Eq. (\ref{eq_Lstate})) with trivial stabilizer (in $G$) fulfills these requirements. We therefore have $\mZ\cap G|\phi\ra=\emptyset$. Moreover, note that both $\mZ$ and $G|\phi\ra$ are $G$-invariant. From the lemma above it follows that there exists a SLIP $f$
such that $f(|\psi\ra)=0$ for all $|\psi\ra\in\mZ$ and $f(|\psi\ra)=1$ for all $|\psi\ra\in G|\phi\ra$.
Now, let $|\chi\ra\in\mH_n$ be another state with $f(|\chi\ra)\neq 0$. Then, $|\chi\ra\notin\mZ$ so that $\dim\left(G|\chi\ra\right)=d$. We assert that the orbit $G|\chi\ra$ is Zariski closed. Otherwise, take $|\xi\ra\in\overline{G|\chi\ra}$ such that $|\xi\ra\notin G|\chi\ra$. Since $f(|\xi\ra)=f(|\chi\ra)\neq 0$ we have $|\xi\ra\neq 0$ and $|\xi\ra\notin\mZ$, and therefore $\dim(G|\xi\ra)=d$.
Now, applying Theorem~6 in page 63 of~\cite{Sh94} with the map $g\mapsto g|\chi\ra$
implies that $G|\chi\ra$ is Zariski open in $\overline{G|\chi\ra}$. Hence, the boundary
$B\eqdef\overline{G|\chi\ra}-G|\chi\ra$ is Zariski closed. Therefore,
since $G|\xi\ra\subset B$ we must have $\dim(G|\xi\ra)<d$ and therefore have a contradiction with the assumption that $G|\chi\ra$ is not Zariski closed.
Hence, for any $|\chi\ra$ with $f(|\chi\ra)\neq 0$, the orbit $G|\chi\ra$ is Zariski closed. We denote by
\be
\mH_f\eqdef\left\{|\psi\ra\in\mH_n\;|\;f(|\psi\ra)\neq 0\right\}\;.
\ee
The set $\mH_f$ is Zariski open since its complement in $\mH_n$ is Zariski closed as it is determined by only one polynomial equation $f(|\psi\ra)=0$. It is not empty since $|\chi\ra\in\mH_f$, and as we argued above, the orbit under $G$ of any state in $\mH_f$ is Zariski closed and of maximal dimension. We therefore conclude that $G\mC$ (the set of all states with closed orbits of maximal dimension, which coincides with ${\cal H}_f$) is Zariski open. Hence, it is open and its complement is of lower dimension in $\mH_n$.
\end{proof}

We now proof Property~3 of Lemma~\ref{properties}.
\begin{theorem}
\label{propertyiii}
The set $\mC$ defined in Eq (\ref{C0}) is a connected smooth submanifold of $\mH_n$, and $K$ acts differentiably on $\mC$.
\end{theorem}

\begin{proof}
We start with the proof the $\mC$ is a submanifold of $\mH_n$. Let $X_{1},...,X_{m}$ be a basis of $Lie(K)$ (for example, the 3 matrices $i\sigma_1$, $-i\sigma_2$, and $i\sigma_3$, where $\{\sigma_j\}_{j=1,2,3}$ are the 3 Pauli matrices, form the basis of the Lie algebra of $SL(2,\mathbb{C})$). Set $f_{j}(\psi)=\left\langle \psi|X_{j}|\psi\right\rangle$. With these notations,
\[
Crit(\mH_n)=\left\{|\psi\ra\in\mH_n \;\Big|\;f_{j}(\psi)=0\quad\forall\;j=1,...,m\right\}.
\]
Since $X_j$ are all anti-Hermitian we have $\operatorname{Re}f_{j}(\psi)=0$  for all $|\psi\ra\in\mH_n$ and all $j=1,...,m$. We therefore define the following symplectic bilinear form $\omega(\psi,\phi)=\operatorname{Im}\left\langle
\psi|\phi\right\rangle$ (the imaginary part). Then
$$
\left( df_{j}\right)_{\psi}(\phi)=\lim_{t\to 0}\frac{d}{dt}f_{j}(\phi+t\psi)=-2i\omega\left(X_{j}\psi,\phi\right)\;.
$$
Since $\omega$ is nondegenerate, we see that if
$\psi\in \mC$  then the set $\{\phi\in\mH_n|\left(df_{j}\right)  _{\psi}(\phi)=0,\;j=1,...,m\}$
has \emph{real} dimension equal to $2^{n+1}-m$ (that is, the real dimension of $\mH_n$ minus the $m$ real (orthogonality) conditions $\operatorname{Im}\la\psi| X_j|\phi\ra=0$). Since $\mC\neq \emptyset$
the implicit function theorem now implies that
$\mC$ is a real submanifold of $\mH_n$ of dimension $2^{n+1}-m$ with $K$
acting differentiably.

Next, we prove that $\mC$ is connected.
Suppose that $\mC=\mC_{1}\cup
\mC_{2}$ with $\mC_{1}$ and $\mC_2$ non-empty sets that are open in $\mC$.
We will show that we get a contradiction if we assume $\mC_{1}\cap
\mC_{2}=\emptyset$.  Then $G\mC_{1}$ and $G\mC_{2}$ are
open and since their union is $G\mC$ (which is Zariski open and dense) they must intersect. Thus there
exist $|\psi\ra\in \mC_{1}$ and $|\phi\ra\in \mC_{2}$ and $g\in G$ such that $g|\psi\ra=|\phi\ra$. Thus
$\left\Vert |\phi\ra\right\Vert =\left\Vert g|\psi\ra\right\Vert \geq\left\Vert |\psi\ra\right\Vert
$ and $\left\Vert |\psi\ra\right\Vert =\left\Vert g^{-1}|\phi\ra\right\Vert \geq\left\Vert
|\phi\ra\right\Vert $ so $\left\Vert |\psi\ra\right\Vert =\left\Vert |\phi\ra\right\Vert $ and this
implies due to the Kempf--Ness theorem that there exists $k\in K$ with $k|\psi\ra=|\phi\ra$. Let $\sigma:[0,1]\rightarrow
K$ be continuous such that $\sigma(0)=I$ and $\sigma(1)=k$ (recall $K$ is connected). Then,
$|\varphi(t)\ra\eqdef \sigma(t)|\psi\ra$  is a continuous curve joining $|\psi\ra$ to $|\phi\ra$ in
$\mC$. This is a contradiction.
\end{proof}

\section*{Appendix C: Proof of Main Theorem}
\label{sectionC}

In this appendix we prove the main result of the paper, namely that almost all qubits states of more than four qubits have a trivial stabilizer (in $\tilde{G}$) (see Theorem~\ref{mtheorem}). In order to do so, we use the principle orbit type theorem and Lemma \ref{properties}, which states some important properties of the set of critical states with finite stabilizer in $G$.
Again, in order to enhance readability we restate the theorem before we prove it.

{\bf Theorem 4.}
{\it For $n\geq 5$ there exists a subset of states $\mA\subset\mH_n$ that is open, dense, and of full measure in $\mH_n$ and for which the stabilizer group $\tilde{G}_{\psi}=\{I\}$ is trivial for all $|\psi\ra\in\mA$.}

{\it Proof of Theorem~\ref{mtheorem}:}
We first show that there exists an open, dense, and of full measure set $\mA'$ in $\mH_n$ such that for all $|\psi\ra\in\mA'$, $G_\psi=\{I\}$ is trivial, and then use it to show that there also exists an open, dense, and of full measure subset $\mA$ in $\mH_n$ for which
$\tilde{G}_\psi=\{I\}$ for all $|\psi\ra\in\mA$.

We start by applying the principle orbit type theorem (that is, Theorem~\ref{principal}) with $\mM$ being the connected smooth submanifold $\mC\subset\mH_n$ (see property~3 of Lemma~\ref{properties}),
and the compact Lie group $C=K$. Since $|L_n\ra\in\mC$ has trivial stabilizer in $K$, the principle orbit type theorem implies that the set of all states $|\psi\ra\in\mC$ with $K_\psi=\{I\}$ is open with complement of lower dimension in $\mC$.
Hence, due to property~1 of Lemma~\ref{properties} we have that the set $\mC'= \left\{|\psi\ra\in\mC\;\big|\;G_\psi=\{I\}\right\}$ is open  with complement of lower dimension in $\mC$.

Next, denote by $G\mC'=\left\{g|\psi\ra\;\big|\;|\psi\ra\in\mC'\;;\;g\in G\right\}$ and note that the stabilizer group of any state $|\phi\ra\in G\mC'$ is trivial. Indeed, let $g\in G$ be such that $|\phi\ra=g|\psi\ra$ with $|\psi\ra\in\mC'$. Then,
$$
G_\phi=G_{g\psi}=gG_\psi g^{-1}=\{I\}\;,
$$
where in the last equality we used the fact that $G_\psi=\{I\}$ since $|\psi\ra\in \mC'$. Now, since $\mC'$ is open with complement of lower dimension in $\mC$, we now show that $G\mC'$ is open with complement of lower dimension in $G\mC$. We will use the symbol $\dim$ for real dimension. Thus $\dim\mathcal{H}%
_{n}=2^{2n}$. If $\left\vert \phi\right\rangle ,\left\vert \psi\right\rangle
\in\mathcal{C}$ then the Kempf-Ness theorem implies that if $g\left\vert
\phi\right\rangle =\left\vert \psi\right\rangle $ then there exists $k\in K$
such that $k\left\vert \phi\right\rangle =\left\vert \psi\right\rangle $. This
implies that $2^{2n}=\dim G\mathcal{C}=\dim G/K+\dim\mathcal{C}$. Arguing in
the same way we see that
\[
\dim(G\mathcal{C}-G\mathcal{C}^{\prime})=\dim G(\mathcal{C}-\mathcal{C}%
^{\prime})=
\]%
\[
\dim G/K+\dim(\mathcal{C}-\mathcal{C}^{\prime})<\dim G/K+\dim\mathcal{C}%
=2^{2n}.
\]

Thus, the set $\mA'\eqdef G\mC'$ is open and has a complement of lower dimension
 in $G\mC$, and since $G\mC$ is open and has a complement of lower dimension
 in $\mH_n$ (see property~2 in Lemma~\ref{properties}), we also have that $\mA'$ is open and has a complement of lower dimension
 in $\mH_n$. Hence, the set of all states in $\mH_n$ with $G_\psi=\{I\}$ is open and has a complement of lower dimension
 in $\mH_n$.

Finally, consider the group $\tilde{G}$. Any element $\tilde{g}\in\tilde{G}$ can be written as $\tilde{g}=tg$ where $0\neq t\in\mathbb{C}$ and $g\in G$. Consider first the case of an even $n$. In this case, define
\be
\mA\eqdef\left\{|\psi\ra\in\mA'\;\Big|\;f_2(|\psi\ra)\neq 0\right\}\;,
\ee
where $f_2$ is the homogeneous SLIP of degree 2 defined in the main text.
Note that the set $\mA$ is open and has a complement of lower dimension in $\mH_n$ since $\mA'$ has these properties.
Now, let $|\psi\ra\in\mA$ and $\tilde{g}\in\tilde{G}_{\psi}$. We then have
\be\label{ggg1}
0\neq f_2\left(|\psi\ra\right)=f_{2}\left(\tilde{g}|\psi\ra\right)=t^2f_{2}\left(g|\psi\ra\right)=t^2f_{2}\left(|\psi\ra\right)\;.
\ee
Hence, $t^2=1$ so that $t=\pm 1$ and $\tilde{g}=\pm g$. Since both $g$ and $-g$ belong to $G$ we conclude that
$\tilde{g}$ must belong to $G$. We therefore have
\be\label{gpsi}
\tilde{G}_{\psi}=G_{\psi}=\{I\}\quad\forall\;|\psi\ra\in\mA\;.
\ee
This completes the proof of the theorem for even number of qubits.
For odd $n$, repeat the same argument as in~\eqref{ggg1} but instead of $f_2$ use the homogeneous SLIP of degree 4 defined in Eq.~\eqref{f4}. In this case we get $t^4=1$ so that $t=\pm 1,\pm i$. In~\cite{GW13} it was shown that for an odd number ($\geq 5$) of qubits there is also a homogeneous SLIP of degree $6$. Denoting this polynomial by $f_6$ we define $\mA$
to be
\be
\mA\eqdef\left\{|\psi\ra\in\mA'\;\Big|\;f_4(|\psi\ra)\neq 0\;;\;f_6(|\psi\ra)\neq 0\right\}\;.
\ee
Again, the set $\mA$ is open and has a complement of lower dimension in $\mH_n$ since $\mA'$ has these properties.
Hence, for all $|\psi\ra\in\mA$ and $\tilde{g}\in\tilde{G}_\psi$ we have $\tilde{g}=tg$ with $t^4=t^6=1$ and $g\in G$.
We therefore must have $t=\pm 1$ so that $\tilde{g}\in G$ and the equalities in Eq.~\eqref{gpsi} holds in this case as well. This completes the proof.$\quad\quad\quad\quad\Box$

\section*{Appendix D: Applications to multipartite entanglement}
\label{sectionD}

We prove here the two theorems presented in the main text, Theorem \ref{ThnoDet} and Theorem \ref{ThmaxProb}.
Whereas the first of these theorems states that there is almost no non--trivial deterministic transformation (via local operations) possible among n--way entangled qubit--states, the second presents the maximal success probability with which the transformation can be achieved. In order to increase readability we restate here the theorems.

{\bf Theorem 6.}
{\it Let $|\psi\ra,|\phi\ra\in\mH_n$ be two $n$-way entangled states and suppose also that the stabilizer $\tilde{G}_\psi=\{I\}$ is trivial. Then, $|\psi\ra$ can be converted deterministically to $|\phi\ra$ by LOCC or SEP operations if and only if $|\psi\ra$ and $|\phi\ra$ are related by a local unitary operation; that is, there exists $u\in \tilde{K}$ such that
$|\psi\ra=u|\phi\ra$.}

\begin{proof}
We prove the theorem for SEP and since LOCC is a subset of SEP the proof applies also to LOCC.
Suppose that there exists a SEP operation $\mE$ such that $|\phi\lr\phi|=\mE(|\psi\lr\psi|)$. Then, $|\psi\ra$ and $|\phi\ra$ must be in the same invertible SLOCC class and there exists $g\in\tilde{G}$ such that $|\phi\ra=g|\psi\ra$.
Note that $g$ is unique since $\tilde{G}_\psi$ is trivial. Let $\{M_j\}$ be the operator sum representation of $\mE$ such that each $M_j$ has a tensor product form $A_1\otimes A_2\otimes\cdots\otimes A_n$. Following similar arguments as in~\cite{GoWa11}, the condition
$|\phi\lr\phi|=\mE(|\psi\lr\psi|)$ is equivalent to
$$
M_j|\psi\ra=c_j|\phi\ra\;,
$$
where $c_j$ are complex numbers satisfying $\sum_j|c_j|^2=1$. Note that if $c_j\neq 0$ then $M_j\in\tilde{G}$ since $\mE$ is SEP and $|\phi\ra$ is $n$-way entangled. Substituting $|\phi\ra=g|\psi\ra$ in the equation above and moving all terms to the LHS we get that for all $j$ such that $c_j\neq 0$ we have
\be
\frac{1}{c_j}g^{-1}M_j|\psi\ra=|\psi\ra\;.
\ee
Now, since the stabilizer $\tilde{G}_\psi=\{I\}$, and $\frac{1}{c_j}g^{-1}M_j\in\tilde{G}$, we get that for
all $j$ such that $c_j\neq 0$
\be
M_j=c_jg\;.
\ee
Combining this with the fact that the transformation is deterministic, i.e. $\sum_jM^{\dag}_{j}M_j=I$ gives
\be
g^\dag g+\sum_{j\;:\;c_j= 0}M_{j}^{\dag}M_j=I\;.
\ee
Acting with both sides of this equation on $|\psi\ra$ gives $g^{\dag}g|\psi\ra=|\psi\ra$, since for $j$ with $c_j=0$ we have $M_j|\psi\ra=0$. Hence, $g^{\dag}g$ is in the stabilizer of $|\psi\ra$. Since the stabilizer is trivial we have $g^{\dag}g=I$
or equivalently, $g\in U$; i.e. $|\phi\ra=g|\psi\ra$ are related by local unitary operation.
\end{proof}

The second theorem, which we prove here determines the maximal success probability with which the transformations studied above are possible.

{\bf Theorem 7.}
{\it Let $|\psi\ra\in\mH_n$ be an $n$-way entangled state with trivial stabilizer $\tilde{G}_\psi=\{I\}$. Let $|\phi\ra=g|\psi\ra$ be a normalized state in the SLOCC orbit $\tilde{G}|\psi\ra$. Then, the maximum probability with which $|\psi\ra$ can be converted to $|\phi\ra$ by LOCC or SEP is given by:
\be\label{pmax}
p_{\max}(|\psi\ra\to|\phi\ra)=\frac{1}{\lambda_{\max}(g^{\dag}g)},
\ee
where $\lambda_{\max}(X)$ denotes the maximal eigenvalue of $X$.}

As mentioned in the main text, the proof below implies that the RHS of~\eqref{pmax} always provides a lower bound on the maximal probability even if $|\psi\ra\notin\mA$. However, if the stabilizer of $|\psi\ra$ is not trivial, $g$ is not unique, and the lower bound given by the RHS of~\eqref{pmax} can be improved by maximizing it over all $g\in\tilde{G}$ that satisfies $|\phi\ra=g|\psi\ra$.

\begin{proof}
We start by showing that this probability can be achieved by LOCC. Since $g\in\tilde{G}$ we can write it as $g=g_1\otimes g_2\otimes\cdots\otimes g_n$ with each $g_j\in GL(2,\mathbb{C})$ and $j=1,...,n$.
Thus, $\lambda_{\max}(g^{\dag}g)=\lambda_1\lambda_2\cdots\lambda_n$, where $\lambda_j\eqdef\lambda_{\max}(g^{\dag}_{j}g_j)$ is the largest eigenvalue of $g_{j}^{\dag}g_j$. For each $j=1,...,n$ we define the following two outcome generalized measurement:
\begin{align}
& N_{0|j}\eqdef \frac{1}{\sqrt{\lambda_j}}g_j\nonumber\\
& N_{1|j}\eqdef \sqrt{I_2-\frac{1}{\lambda_j}g_{j}^{\dag}g_j},
\end{align}
where $I_2$ is the $2\times 2$ identity matrix. Note that from the definition of $\lambda_j$ we have $I_2-\frac{1}{\lambda_j}g_{j}^{\dag}g_j\geq 0$. Thus,
the matrices above define a local generalized measurement since $N_{0|j}^{\dag}N_{0|j}+N_{1|j}^{\dag}N_{1|j}=I_2$.
Now, suppose that each of the $n$ parties performs this measurement on their $j$'s share of $|\psi\ra$. Then, the probability that each of the $n$ parties get the outcome $0$ is:
\begin{align}
& \la\psi|N_{0|1}^{\dag}N_{0|1}\otimes N_{0|2}^{\dag}N_{0|2}\otimes\cdots\otimes N_{0|n}^{\dag}N_{0|n}|\psi\ra\nonumber\\
& =\frac{1}{\lambda_1\cdots\lambda_n}\la\psi|g^{\dag}g|\psi\ra=\frac{1}{\lambda_{\max}(g^{\dag}g)}
\end{align}
where in the last equality we have used the fact that $|\phi\ra=g|\psi\ra$ is normalized. Furthermore, the state after this outcome occur is
\be
\frac{1}{\|N_{0|1}\otimes\cdots\otimes N_{0|n}|\psi\ra\|}N_{0|1}\otimes\cdots\otimes N_{0|n}|\psi\ra=|\phi\ra\;.
\ee
This completes the proof that $p_{\max}$ in~\eqref{pmax} is achievable. We now show that it is also optimal.

Any LOCC (or SEP) protocol that convert $|\psi\ra$ to $|\phi\ra$ with probability $p$, also convert $|\psi\ra$ to other states with probability $1-p$. These other states can always be converted to the product state $|00...0\ra$. It is therefore enough to consider LOCC protocols converting $|\psi\ra$ to $|\phi\ra$ with probability $p$ and converting $|\psi\ra$ to $|00...0\ra$ with probability $1-p$.

Let $\mE=\mE_0+\mE_1$ be an LOCC CPTP map with $\mE_0$ and $\mE_1$ being CP and trace decreasing maps, satisfying
\begin{align}\label{req}
\mE_0(|\psi\lr\psi|) & =(1-p)|0...0\lr 0...0|\nonumber\\
\mE_1(|\psi\lr\psi|) & =p|\phi\lr\phi|
\end{align}
We denote the operator sum representations of $\mE_0$ and $\mE_1$ by $\{M_i\}$ and $\{N_{j}\}$, respectively. We therefore have
\be
M_i|\psi\ra=a_i|0...0\ra\text{ and }N_{j}|\psi\ra=b_j|\phi\ra\;,
\ee
with $\sum_i|a_i|^2=1-p$ and $\sum_j|b_j|^2=p$.
Now, if $b_j\neq 0$ we get that $\frac{1}{b_j}N_j|\psi\ra=g|\psi\ra$ so that $\frac{1}{b_j}g^{-1}N_j\in\tilde{G}_{\psi}$.
But since the stabilizer $\tilde{G}_\psi$ is trivial we must have
for all $j$ with $b_j\neq 0$:
\be
N_j=b_jg\;.
\ee
Combining this with the fact that $\mE$ is trace preserving gives
\be
pg^{\dag}g+\sum_{j\;:\;b_j= 0}N_{j}^{\dag}N_j+\sum_{i}M_{i}^{\dag}M_{i}=I\;.
\ee
This implies that $I-pg^{\dag}g$ must be a positive semi-definite matrix. Hence, $p\leq 1/\lambda_{\max}(g^{\dag}g)$ as asserted.
\end{proof}

\end{appendix}
\end{document}